%% file: main.tex
\documentclass[runningheads]{llncs}

\usepackage[hidelinks]{hyperref}
\usepackage{graphicx}
\usepackage{cite}

\usepackage{enumitem}

\usepackage{ amsmath }

\usepackage{ amsfonts }

\usepackage{ stmaryrd }

\usepackage{float}

\begin{document}
\title{NP Satisfiability for Arrays as Powers}
\titlerunning{NP Satisfiability for Arrays as Powers}
%
\author{Rodrigo Raya\orcidID{0000-0002-0866-9257} \and
Viktor Kun\v{c}ak\orcidID{0000-0001-7044-9522}}
%
%
\institute{School of Computer and Communication Science \\ École Polytechnique Fédérale de Lausanne (EPFL), Switzerland 
\email{\{rodrigo.raya,viktor.kuncak\}@epfl.ch}}
\maketitle              
\begin{abstract}
We show that the satisfiability problem for the quantifier-free theory of product structures with the equicardinality relation is in NP. As an application, we extend the combinatory array logic fragment to handle cardinality constraints. The resulting fragment is independent of the base element and index set theories.
\end{abstract}
\section{Introduction}
\input{introduction}

\section{NP Complexity for Power Structures}
\input{power}

\section{Explicit Sets of Indices and a Polynomial Verifier for QFBAPA}
\input{qfbapa}

\section{NP Complexity for QFBAPAI}
\input{qfbapai}

\section{Combination with the Array Theory}
\input{combination}

\section{Related Work}
\input{related}

\section{Conclusion and Future Work}
\input{conclusion}

%
%
%

\bibliographystyle{splncs04}
{\raggedright 
\bibliography{references}
}
\end{document}

%% file: introduction.tex
Arrays are a fundamental data structure in computer
science. Decision procedures for arrays are therefore of paramount importance for deductive program verification. A number of results have examined fragments that strike interesting trade-offs between expressive power and complexity \cite{mccarthy_towards_1993, stump_decision_2001, bradley_whats_2006, daca_array_2016, habermehl_what_2008, ghilardi_decision_2007}.

 A particularly important fragment for formal verification is combinatory array logic (CAL) fragment \cite{de_moura_generalized_2009}, which is implemented in the widely used Z3 theorem prover \cite{hutchison_z3_2008}.
A key to expressive power of the generalized array fragment is that it extends the extensional quantifier-free theory of arrays \cite{stump_decision_2001} with point-wise functions and relations. 

In this paper, we start by observing that the generalized array fragment signature corresponds to the signature of a product structure \cite{hodges_model_1993}. The decidability of product structures has been studied in the literature on model theory \cite{mostowski_direct_1952, feferman_first_1959}. Moreover, these results also cover formulas that constrain sets of indices using, for example, equicardinality relation \cite{feferman_first_1959}, which provides additional expressive power. Unfortunately, the results from model theory typically consider quantified first-order theory, resulting in high complexity \cite{ferrante_computational_1979} even when instantiated to the case of no quantifier alternations. The basic source of this inefficiency is that the underlying procedure explicitly constructs exponentially many formulas.

On the other hand, the result (theorem 17 of \cite{de_moura_generalized_2009})  implies that the satisfiability problem of the quantifier-free theory of a power structure is in NP whenever the theory of the components is. 

In this paper, we present a direct proof of the NP membership for satisfiability of formulas in power structures. The proof is independent of the theories of the indices and the theory of array elements. As a consequence, we obtain that the satisfiability problem of the quantifier-free fragment of Skolem arithmetic is in NP \cite{gradel_dominoes_1989}, which, interestingly, was shown using results in number theory.

As a main contribution, we generalize our construction to prove that the satisfiability problem of the quantifier-free fragment of BAPA \cite{kuncak_deciding_2006} is in NP when set variables are interpreted with index sets defined by formulas of the language of the component theory. Whereas the quantifier-free fragment of BAPA (termed QFBAPA) was shown to be in NP \cite{kuncak_towards_2007}, it was not clear that such construction carries over to the situation where index sets are \emph{interpreted} to be positions in the arrays.

In this paper we show that interpreting QFBAPA sets as sets of array indices that satisfy certain formula results in a logic whose satisfiability is still in NP. We call this new quantifier-free theory QFBAPAI. We show how to use it to encode constraints that mimic those of combinatory array logic \cite{de_moura_generalized_2009}. The result is a logic that can express cardinality constraints that hold componentwise. Unlike \cite{daca_array_2016}, the logic is independent of the component or the index theory. Our formalism shows that QFBAPA sets can be interpreted, overcoming a limitation pointed out in \cite{alberti_cardinality_2017}.

%% file: power.tex
Throughout the paper, we fix a first-order language $L$, a non-empty set $I$ and a structure $\mathcal{M}$ with carrier $M$ for the components of the arrays. We model arrays as a particular kind of product structure:

\begin{definition}
The power structure $\Pi$ has the function space $M^I$ as domain and interprets the symbols of the language $L$ as follows:

\begin{itemize}
\item For each constant $c$ and $i \in I$, $c^\Pi(i) = c^{\mathcal{M}}$.
\item For each function symbol $f$, $i \in I$, $n \in \mathbb{N}$ and $ (a_1,\ldots,a_n) \in (M^I)^n$: 
\[
f^\Pi(a_1,\ldots,a_n)(i) = f^{\mathcal{M}}(a_1(i),\ldots,a_n(i))
\]
\item For each relation symbol $R$, $n \in \mathbb{N}$ and $(a_1,\ldots,a_n) \in (M^I)^n$:
\[
(a_1,\ldots,a_n) \in R^\Pi \text{ if and only if for every } i \in I,
(a_1(i),\ldots,a_n(i)) \in R^{\mathcal{M}}
\]
\end{itemize}
\end{definition}

\noindent We will write tuples $(a_1, \ldots, a_n) \in (M^I)^n$ as $\overline{a}$ and $(a_1(i),\ldots,a_n(i))$ as $\overline{a}(i)$.

\begin{definition}
The quantifier-free theory of a model $\mathcal{N}$, $Th_{\exists^*}(\mathcal{N})$, is the set of existentially quantified formulas $\varphi$ of $L$ such that $\mathcal{N} \models \varphi$. A solution of the formula is a satisfying assignment to the existential variables.
\end{definition}

\begin{lemma} \label{lem:size}
Let $\psi$ be a first-order formula in prenex form and $C$ a disjunct of the DNF form of its matrix. Then $|C| = O(|\psi|)$.
\end{lemma}
\begin{proof}
The DNF conversion only affects the propositional structure of the formula. Thus, in $C$ one may at most have the relations occurring in $\psi$ and their negations. In the worst case, one gets at most $2 |\psi|$ symbols accounting for the relations and at most $4 |\psi|$ symbols accounting for the conjunctions and negations. Therefore, $|C| \le 6 \cdot |\psi|$.
\end{proof}

The following result shows the spirit of our complexity analysis: we take a classical construction (power structure) but
analyze its complexity for quantifier-free fragment that is relevant for program verification.

\begin{theorem} \label{thm:power}
$Th_{\exists^*}(\mathcal{M}) \in \mbox{NP}$ if and only if $Th_{\exists^*}(\Pi) \in \mbox{NP}$.
\end{theorem}
\begin{proof}
Assume that $V_C$ is a polynomial time verifier for $Th_{\exists^*}(\mathcal{M})$. Figure \ref{fig:prod-verifier} gives a polynomial time verifier $V$ for $Th_{\exists^*}(\Pi)$. We show that the machine is a verifier for $Th_{\exists^*}(\Pi)$: 

\begin{figure*}[ht!]
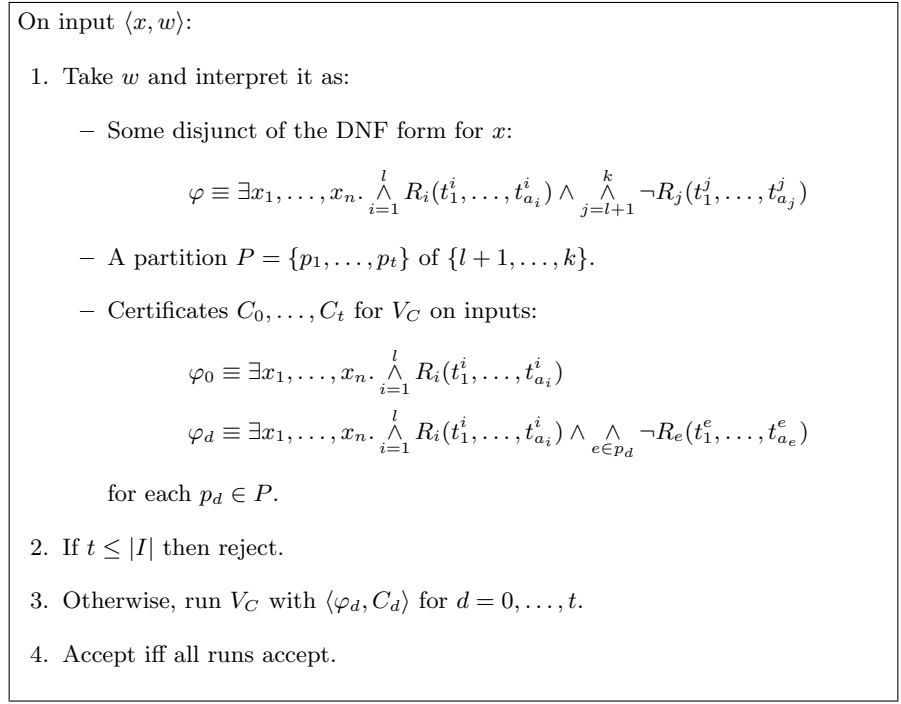

\fbox{\parbox{.95\textwidth}{
On input $\langle x, w \rangle$:

\begin{enumerate}
\setlength\itemsep{1em}
\item Take $w$ and interpret it as:
    \vspace{1em}
    \begin{itemize}
        \item Some disjunct of the DNF form for $x$: \[\varphi \equiv \exists x_1, \ldots, x_n. \mathop{\land}_{i = 1}^l R_i(t_1^i, \ldots, t_{a_i}^i) \land \mathop{\land}_{j 
= l+1}^k \lnot R_j(t_1^j, \ldots, t_{a_j}^j)\]

        \item  A partition $P = \{p_1, \ldots, p_t\}$ of $\{l+1,\ldots,k\}$.

        \item Certificates $C_0, \ldots, C_t$ for $V_C$ on inputs: \begin{align*}
\varphi_0 & \equiv  \exists x_1, \ldots, x_n. \mathop{\land}_{i = 1}^l R_i(t_1^i, \ldots, t_{a_i}^i) \\ \varphi_{d} & \equiv \exists x_1, \ldots, x_n.  \mathop{\land}_{i = 1}^l R_i(t_1^i, \ldots, t_{a_i}^i) \land \mathop{\land}_{e \in p_d} \lnot R_e(t_1^e, \ldots, t_{a_e}^e)    
\end{align*}

    for each $p_d \in P$.
    \end{itemize}
     
\item If $t \le |I|$ then reject.

\item Otherwise, run $V_C$ with $\langle \varphi_d, C_d \rangle$ for $d = 0, \ldots, t$.

\item Accept iff all runs accept.
\end{enumerate}
}}
\caption{Verifier for $Th_{\exists^*}(M^I)$}
\label{fig:prod-verifier}
\end{figure*}

\vspace{.5em}

\begin{itemize}
\setlength\itemsep{.5em}
\item $w$ has polynomial size in $|x|$:

\begin{itemize}[label={}]
\vspace{.5em}

\item By lemma \ref{lem:size}, $|\varphi| = O(|x|)$.

\item Thus, $k = O(|x|)$.

\item $P = O(|x|^2)$ since $P$ can be written with $k \log(k) + k$ bits. 

\item Since $|C_d| = O(|\varphi_d|^{c_d})$ and $|\varphi_d| \le |\varphi| = O(|x|)$, $|C_d| = O(|x|^{c_d})$.
\item Thus, $|w| = |\varphi| + |P| + \sum\limits_{d = 0, \ldots, t} |C_d| = O\Big(|x|^{\max\{2, \max\limits_d{c_d}\}}\Big)$.
\end{itemize}

\item $V$ runs in polynomial time in $|x|$: 

\vspace{.5em}

\begin{itemize}[label={}]
\item Building the list of $\varphi_d$ is $O(|x|^2)$.

\item As above, $|\varphi_d| \le |\varphi| = 
O(|x|)$.

\item So each call to $V_C$ runs in $O(|x|^f)$ ($V_C$ is polynomial time).

\item Like before, $k = O(|x|)$.

\item Therefore, $V$ runs in $O\Big(|x|^{\max\{2,f + 1\}}\Big)$.
\end{itemize}

\item $V$ is a verifier for $Th_{\exists^*}(\Pi)$:

\vspace{1em}

$\Rightarrow)$ If $x \in Th_{\exists^*}(\Pi)$ then writing $x$ in prenex DNF form, there is at least one disjunct $\varphi$ (as in figure \ref{fig:prod-verifier}) true in the product. Thus, there is $\overline{s} \in M^I$ satisfying:  \[
\begin{array}{c}
\dotfill \\

({t_1^i}^{\Pi}[\overline{x} \mapsto \overline{s}], \ldots, {t_{a_i}^i}^{\Pi}[\overline{x} \mapsto \overline{s}]) \in R_i^\Pi \\

\dotfill \\

({t_1^j}^{\Pi}[\overline{x} \mapsto \overline{s}], \ldots, {t_{a_j}^j}^{\Pi}[\overline{x} \mapsto \overline{s}]) \notin R_j^\Pi \\

\dotfill
\end{array}
\]

Using the semantics of products this means: \[
\begin{array}{c}
\dotfill \\

\forall r \in I. ({t_1^i}^{\mathcal{M}}[\overline{x} \mapsto \overline{s}(r)], \ldots, {t_{a_i}^i}^{\mathcal{M}}[\overline{x} \mapsto \overline{s}(r)]) \in R_i^{\mathcal{M}} \\

\dotfill \\

\exists r \in I. ({t_1^j}^{\mathcal{M}}[\overline{x} \mapsto \overline{s}(r)], \ldots, {t_{a_j}^j}^{\mathcal{M}}[\overline{x} \mapsto \overline{s}(r)]) \notin R_j^\mathcal{M} \\

\dotfill
\end{array}
\]

So there is a map $r: \{l+1, \ldots, k\} \to I$ that assigns to each formula, one index where it holds. $r$ induces a partition $P = r^{-1}(I)$ of $\{l+1, \ldots, k\}$ with $t = |P| \le \min(|I|,k-l)$. Each part $p_d = \{e_1, \ldots, e_m\}$ and each associated index $r_d = r(e_i)$, satisfy the following system:
\[
\begin{array}{c}
\dotfill \\

({t_1^i}^{\mathcal{M}}[\overline{x} \mapsto \overline{s}(r_p)], \ldots, {t_{a_i}^i}^{\mathcal{M}}[\overline{x} \mapsto \overline{s}(r_p)]) \in R_i^\mathcal{M} \\

\dotfill \\

({t_1^{e_1}}^{\mathcal{M}}[\overline{x} \mapsto \overline{s}(r_p)], \ldots, {t_{a_{e_1}}^{e_1}}^{\mathcal{M}}[\overline{x} \mapsto \overline{s}(r_p)]) \notin R_{e_1}^\mathcal{M} \\

\dotfill \\

({t_1^{e_m}}^{\mathcal{M}}[\overline{x} \mapsto \overline{s}(r_p)], \ldots, {t_{a_{e_m}}^{e_m}}^{\mathcal{M}}[\overline{x} \mapsto \overline{s}(r_p)]) \notin R_{e_m}^\mathcal{M}
\end{array}
\] 
Equivalently, for each $d \in \{ 1, \ldots, t\}$, $
\mathcal{M} \models \varphi_d[\overline{x} \mapsto \overline{s}(r_d)]
$. For $d = 0$, we set:
\[
r_0 = 
\begin{cases}
\text{any index } i \in I & \text{ if } t = 0 \\
\text{some } r_d \in \{r_1, \ldots, r_t\}         & \text{ if } t > 0
\end{cases}
\] Then $
\mathcal{M} \models \varphi_0[\overline{x} \mapsto \overline{s}(r_0)]
$. By definition of $V_C$, there are polynomially-sized certificates $C_0, \ldots, C_t$ such that $V_C$ accepts $\langle \varphi_d, C_d \rangle$ for each $d$. Thus $V$ accepts $\langle x, \langle \varphi, P, C_0, \ldots, C_t \rangle \rangle$.

\vspace{1em}

$\Leftarrow)$ Let $w = \langle \varphi, P, \{C_d\}_{d \in \{0, \ldots, t\}} \rangle$ be a certificate such that $V$ accepts $\langle x, w \rangle$. Then, by step 2, $t = |P| \le |I|$ and for each $d \in \{0, \ldots, t\}$, $V_C$ accepts $\langle \varphi_d, C_d \rangle$, i.e. $\mathcal{M} \models \varphi_d$. So there are solutions $x_{\cdot i} = (x_{1i}, \ldots, x_{ni})^t$ to the formulas: \begin{align*}
\varphi_0 & \equiv  \exists x_{10}, \ldots, \exists x_{n0}. \mathop{\land}_{i = 1}^l R_i(t_1^i, \ldots, t_{a_i}^i) \\
\varphi_d & \equiv  \exists x_{1d}, \ldots, \exists x_{nd}. \mathop{\land}_{i = 1}^l R_i(t_1^i, \ldots, t_{s_i}^i) \land \mathop{\land}_{e \in p_d} \lnot R_e(t_1^e, \ldots, t_{a_e}^e)     
\end{align*}

Fix distinct $i_1, \ldots, i_t \in I$. Consider the $n \times |I|$ matrix with entries:
\[
s_{ji} = 
\begin{cases}
x_{ji} & \text{if } i \in \{i_1, \ldots, i_t\} \\
x_{j0} & \text{otherwise}
\end{cases}
\]
The rows of this matrix $\overline{s} = \{s_1, \ldots, s_n\}$ are solutions of $\varphi$ in the product structure:
\[
\begin{array}{c}
\dotfill \\
({t_1^i}^{\Pi}[\overline{x} \mapsto \overline{s}], \ldots, {t_{a_i}^i}^{\Pi}[\overline{x} \mapsto \overline{s}]) \in R_i^\Pi \\

\dotfill \\

({t_1^j}^{\Pi}[\overline{x} \mapsto \overline{s}], \ldots, {t_{a_j}^j}^{\Pi}[\overline{x} \mapsto \overline{s}]) \notin R_j^\Pi \\

\dotfill
\end{array}
\]

Using the definition of product, it is sufficient to show: \[
\begin{array}{c}
\dotfill \\
\forall r \in I. ({t_1^i}^{\mathcal{M}}[\overline{x} \mapsto s(r)], \ldots, {t_{a_i}^i}^{\mathcal{M}}[\overline{x} \mapsto s(r)]) \in R_i^\mathcal{M} \\

\dotfill \\

\exists r \in I. ({t_1^j}^{\mathcal{M}}[\overline{x} \mapsto s(r)], \ldots, {t_{a_j}^j}^{\mathcal{M}}[\overline{x} \mapsto s(r)]) \notin R_j^\mathcal{M} \\

\dotfill
\end{array}
\]

For $i \in \{1, \ldots, l\}$ and each $r \in I$, the following formula needs to hold: \[
({t_1^i}^{\mathcal{M}}[\overline{x} \mapsto s(r)], \ldots, {t_{s_i}^i}^{\mathcal{M}}[\overline{x} \mapsto s(r)]) \in R_i^\mathcal{M}
\] If $r \in \{i_1, \ldots, i_t\}$ then $s(r) = x_{\cdot r}$ (i.e. all $x_{1r}, \ldots, x_{nr}$) and the equation holds since $\mathcal{M} \models \varphi_r[x_{\cdot r}]$. Otherwise, $s(r) = x_{\cdot 0}$ and the equation holds since $\mathcal{M} \models \varphi_0[x_{\cdot 0}]$.

\vspace{1em}

For $j \in \{l+1, \ldots, k\}$ and some $r \in I$, the following formula needs to hold: \[ ({t_1^j}^{\mathcal{M}}[\overline{x} \mapsto s(r)], \ldots, {t_{s_j}^j}^{\mathcal{M}}[\overline{x} \mapsto s(r)]) \notin R_j^\mathcal{M}
\] We take $r = i_d$ such that $j \in p_d$. Then $s(r) = x_{\cdot r}$ and the equation holds since $\mathcal{M} \models \varphi_r[x_{\cdot r}]$.
\end{itemize}

This proves the left to right implication. For completeness, we sketch a justification of the intuitively clear right to left implication. The idea is that one can extend the signature of $L$ with relations $R$ whose interpretation is that of any quantifier-free formula $\varphi$ while retaining NP complexity. Indeed, let $\mathcal{N}$ be any structure for the language $L$ and let $\varphi(x_1,\ldots,x_n)$ be any formula of $L$. Define $R(x_1, \ldots,x_n) := \varphi(x_1,\ldots,x_n)$ and $\mathcal{N}^{ext}$ the model $\mathcal{N}$ extended with the relation symbol $R$ in such a way that $R^{\mathcal{N}^{ext}}(v_1,\ldots,v
_n) = \varphi^{\mathcal{N}}(v_1, \ldots,v_n)$ for values $v_i$ of the carrier of $\mathcal{N}$. We show that:
\[
Th_{\exists^*}(\mathcal{N}) \in NP \iff Th_{\exists^*}(\mathcal{N}^{ext}) \in NP
\]
First observe that $|\varphi(x_1,\ldots,x_n)|$ is an affine function in $|x_i|$: there is a constant term accounting for the logical symbols, plus terms $a_i |x_i|$ accounting for the occurrences of the $x_i$. Now, if $\psi \in Th_{\exists^*}(\mathcal{N})$ then when we contract the occurrences of $\varphi$ into $R$ we still get a linear size in $|\psi|$. Therefore, the verifier for $Th_{\exists^*}(\mathcal{N}^{ext})$ gives the result. If on the other hand, $\psi \in Th_{\exists^*}(\mathcal{N}^{ext})$ then expanding the occurrences of $\mathcal{R}$ each $|x_i|$ is bounded in $|\psi|$, so the expanded expression augments its size by a quadratic factor $O(|\psi|^2)$. The verifier for $Th_{\exists^*}(\mathcal{N})$ gives the result. Finally, let's see that:
\[
Th_{\exists^*}(\Pi^{ext}) \in NP \implies Th_{\exists^*}(\mathcal{M}) \in NP
\]
Given $\varphi \in Th_{\exists^*}(\Pi^{ext})$, we define a relation $R := \varphi$ and consider the corresponding extended language $Th_{\exists^*}(\Pi^{ext(\varphi)})$ which by assumption is in NP. Thus, it is decidable in NP that $R$ holds in the product structure. But, $R^\Pi \equiv \forall i \in I. \varphi^{\mathcal{M}}$. Given that $I$ is non-empty, we have that the verifier for $Th_{\exists^*}(\Pi^{ext(\varphi)})$ can determine if $\varphi \in Th_{\exists^*}(\mathcal{M})$.
\end{proof}

\subsection{Corollary: Quantifier-Free Skolem Arithmetic is in NP}

Although not needed for our final result, the technique of theorem \ref{thm:power} is of independent interest. An example is showing that the satisfiability problem for the quantifier-free fragment of Skolem arithmetic is in NP. This result was first proved by Gr\"adel \cite{gradel_dominoes_1989} using results by Sieveking and von zur Gathen \cite{gathen_bound_1978} with a proof that appears, on the surface, to be specific to the arithmetic theories.

Skolem arithmetic is the first-order theory of the structure $\langle \mathbb{N} \setminus \{0\}, \cdot, =, | \rangle$. This structure is isomorphic to the weak direct power \cite{feferman_first_1959,ferrante_computational_1979} of the structure $\langle \mathbb{N}, +, \le \rangle$. Thus, their existential theories coincide. A variation of the verifier in figure \ref{fig:prod-verifier}, ensuring that if $|I|$ is infinite then $0^n$ is a solution of $\varphi_0$ in $\mathcal{M}$, yields the NP complexity bound for the satisfiability of (existential and) quantifier-free formulas.

%% file: qfbapa.tex
To prepare for generalization of the result from the previous section, we now review the QFBAPA complexity \cite{kuncak_towards_2007} using the notation of the present paper.
The intuition for our approach is that the verifier of figure \ref{fig:prod-verifier} is solving constraints on the array indices which can be schematically presented as in figure \ref{fig:indexset}. The figure presents a Venn region of sets defined by formulas of L. All indices must remain within the boundaries of the main region $A$. This region corresponds to the positive literals of $\varphi$ in figure \ref{fig:prod-verifier}: $\mathop{\land}\limits_{i = 1}^l R_i(t_1^i,\ldots,t_{a_i}^i)$. The negative literals $\mathop{\land}\limits_{j = l+1}^k \lnot R_j(t_1^j, \ldots, t_{a_j}^j)$ generate existential constraints. These can be interpreted as requiring a cardinality greater or equal than one in certain subregions of $A$. 

\begin{figure}[ht!]
\centering
\includegraphics[scale=0.7]{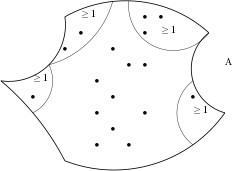}
\caption{An example Venn region with product constraints.}
\label{fig:indexset}
\end{figure}

To generalize our result we use the logic BAPA \cite{kuncak_deciding_2006}, whose language allows to express boolean algebra and cardinality constraints on sets. The satisfiability problem for the quantifier-free fragment of BAPA, often written as QFBAPA, is in NP (see section 3 of \cite{kuncak_towards_2007}). Figure \ref{fig:qfbapa-syntax} shows the syntax of the fragment. $F$ presents the boolean structure of the formula, $A$ stands for the top-level constraints, $B$ gives the boolean restrictions and $T$ the Presburger arithmetic terms. $\mathcal{U}$ represents the universal set and $MAXC$ gives the cardinality of $\mathcal{U}$. We will assume this cardinality to be finite for simplicity of the presentation, but it is straightforward to generalize the NP membership result to the case where the universe is infinite and the language contains additional predicate expressing finiteness of a set \cite[Section 3]{kuncak_ordered_2010}.
\begin{figure}[H]
\centering
\begin{align*}
F & ::= A \, | \, F_1 \land F_2 \, | \, F_1 \lor F_2 \, | \, \lnot F \\
A & ::= B_1 = B_2 \, | \, B_1 \subseteq B_2 \, | \, T_1 = T_2 \, | \, T_1 \le T_2 \, | \, K \text{ dvd } T \\
B & ::= x \, | \, \emptyset \, | \, \mathcal{U} \, | \, B_1 \cup B_2 \, | \, B_1 \cap B_2 \, | \, B^c \\
T & ::= k \, | \, K \, | \, \text{MAXC} \, | \, T_1 + T_2 \, | \, K \cdot T \, |  \, |B| \\
K & ::= \ldots \, | \, -2 \, | \, -1 \, | \, 0 \, | \, 1 \, | \, 2 \, | \, \ldots
\end{align*}
\caption{QFBAPA's syntax}
\label{fig:qfbapa-syntax}
\end{figure}

The basic argument to establish NP complexity of QFBAPA is based on a theorem by Eisenbrand and Shmonin \cite{eisenbrand_carathxe9odory_2006}, which in our context says that any element of an integer cone can be expressed in terms of a polynomial number of generators. Figure \ref{fig:pa-verifier} gives a verifier for this basic version of the algorithm.

The key step is showing equisatisfiability between 2.(b) and 2.(c). If $x_1, \ldots, x_k$ are the variables occurring in $b_0, \ldots, b_p$ then we write $p_\beta = \bigcap\limits_{i = 1}^k x_i^{e_i}$ for $\beta = (e_1,\ldots,e_k)$, $l_\beta = |p_\beta|$ and $\llbracket b_i \rrbracket_{\beta_j}$ the evaluation of $b_i$ as a propositional formula with the assignment given in $\beta$. Now, $|b_i| =  \sum\limits_{j = 0}^{2^e-1} \llbracket b_i \rrbracket_{\beta_j} l_{\beta_j}$, so the restriction $\bigwedge\limits_{i = 0}^k |b_i| = k_i$ becomes $\bigwedge\limits_{i = 0}^p \sum\limits_{j = 0}^{2^e-1} \llbracket b_i \rrbracket_{\beta_j} l_{\beta_j} = k_i$ which can be seen as a linear combination in $\{(\llbracket b_0 \rrbracket_{\beta_j}, \ldots, \llbracket b_p \rrbracket_{\beta_j}). j \in \{0, \ldots, 2^e-1\}\}$. Eisenbrand-Shmonin's result allows then to derive 2.(c) for $N$ polynomial in $|x|$. In the other direction, it is sufficient to set $l_{\beta_j} = 0$ for $j \in \{0, \ldots, 2^e-1\} \setminus \{i_1, \ldots, i_N\}$. Thus, we have:

\begin{theorem}[\cite{kuncak_towards_2007} ] \label{qfbapa-complexity}
QFBAPA is in NP. 
\end{theorem}

\begin{figure*}[ht!]
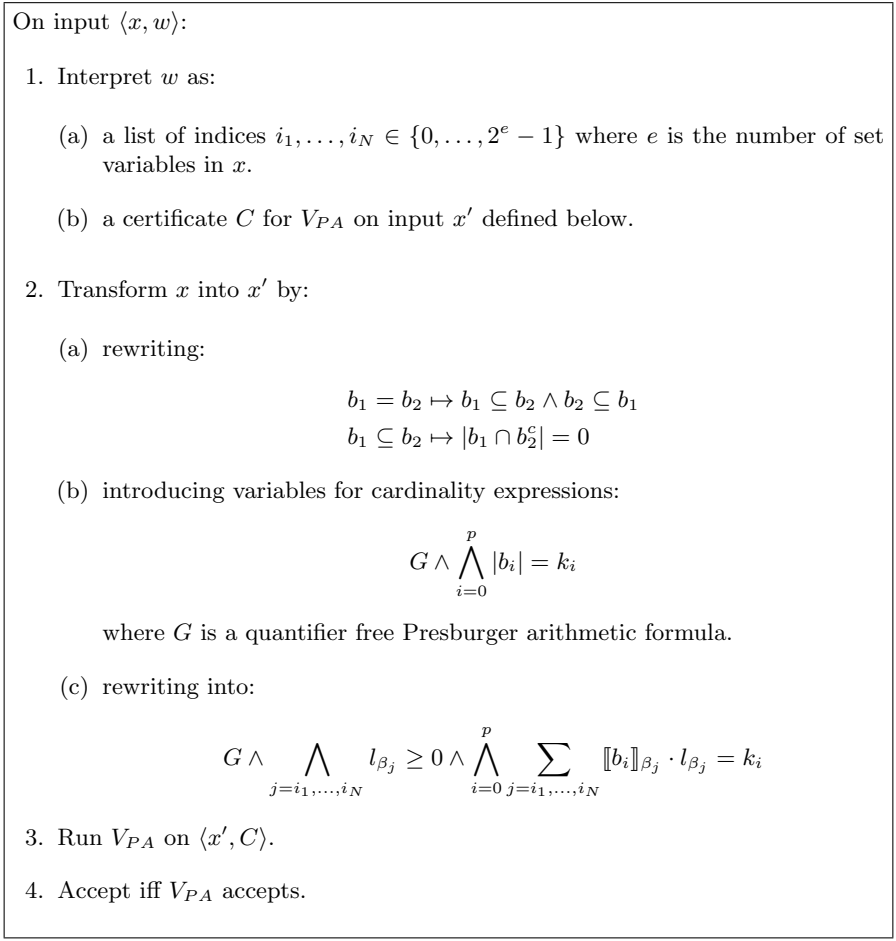

\fbox{\parbox{.95\textwidth}{
On input $\langle x, w \rangle$:

\begin{enumerate}
\setlength\itemsep{1em}
\item Interpret $w$ as:

\vspace{1em}

\begin{enumerate}
    \item a list of indices $i_1, \ldots, i_N \in \{0, \ldots, 2^e-1 \}$ where $e$ is the number of set variables in $x$.
    \item a certificate $C$ for $V_{PA}$ on input $x'$ defined below.
\end{enumerate}

\vspace{.5em}

\item Transform $x$ into $x'$ by:

\vspace{1em}

\begin{enumerate}
    \item rewriting: \begin{align*}
    b_1 = b_2 & \mapsto b_1 \subseteq b_2 \land b_2 \subseteq b_1 \\
    b_1 \subseteq b_2 & \mapsto |b_1 \cap b_2^c| = 0
    \end{align*}
    
    \item introducing variables for cardinality expressions: $$G \land \bigwedge_{i = 0}^{p} |b_i| = k_i$$ where $G$ is a quantifier free Presburger arithmetic formula.
    
    \item rewriting into:
    $$ G \land \bigwedge\limits_{j = i_1, \ldots, i_N} l_{\beta_j} \ge 0 \land \bigwedge_{i = 0}^{p} \sum_{j = i_1, \ldots, i_N} \llbracket b_i \rrbracket_{\beta_j} \cdot l_{\beta_j} = k_i$$
\end{enumerate}

\item Run $V_{PA}$ on $\langle x', C \rangle$.

\item Accept iff $V_{PA}$ accepts.
\end{enumerate}
}}
\caption{Verifier for QFBAPA}
\label{fig:pa-verifier}
\end{figure*}

%% file: qfbapai.tex
We are now ready to present our main result, which extends NP membership of product structures and of QFBAPA to the situation where we interpret QFBAPA sets as sets of indices (subsets of the set $I$) in which quantifier-free formulas hold.

\begin{definition}
\noindent We consider the satisfiability problem for QFBAPA formulas $F$ whose set variables are index sets defined by quantifier-free formulas $\varphi_i$ of $L$ applied to either component theory constants or to components of the array variables:
\[  
\begin{split}
\exists c_1, \ldots, c_m.
& \exists x_1, \ldots, x_n. \\ 
     & F(S_1,\ldots,S_k) \land \bigwedge_{i = 1}^k 
       S_i = \{ r \in I \mid \varphi_i(x_1(r),\ldots,x_n(r), c_1, \ldots, c_m) \}
\end{split}       
\]

\noindent We call this problem \mbox{QFBAPAI}, standing for interpreted \mbox{QFBAPA}.
\end{definition}

\begin{theorem} \label{thm:qfbapai}
QFBAPAI is in NP.
\end{theorem}
\begin{proof}

Let $V_{QFBAPA}$ be a polynomial time verifier for QFBAPA and let $V_C$ be a polynomial time verifier for the component theory. Figure \ref{fig:intverif} gives a verifier $V$ for QFBAPAI. We abbreviate $(x_1, \ldots, x_n)$ by $\overline{x}$ and $(c_1, \ldots, c_m)$ by $\overline{c}$.

\begin{figure*}[ht!]
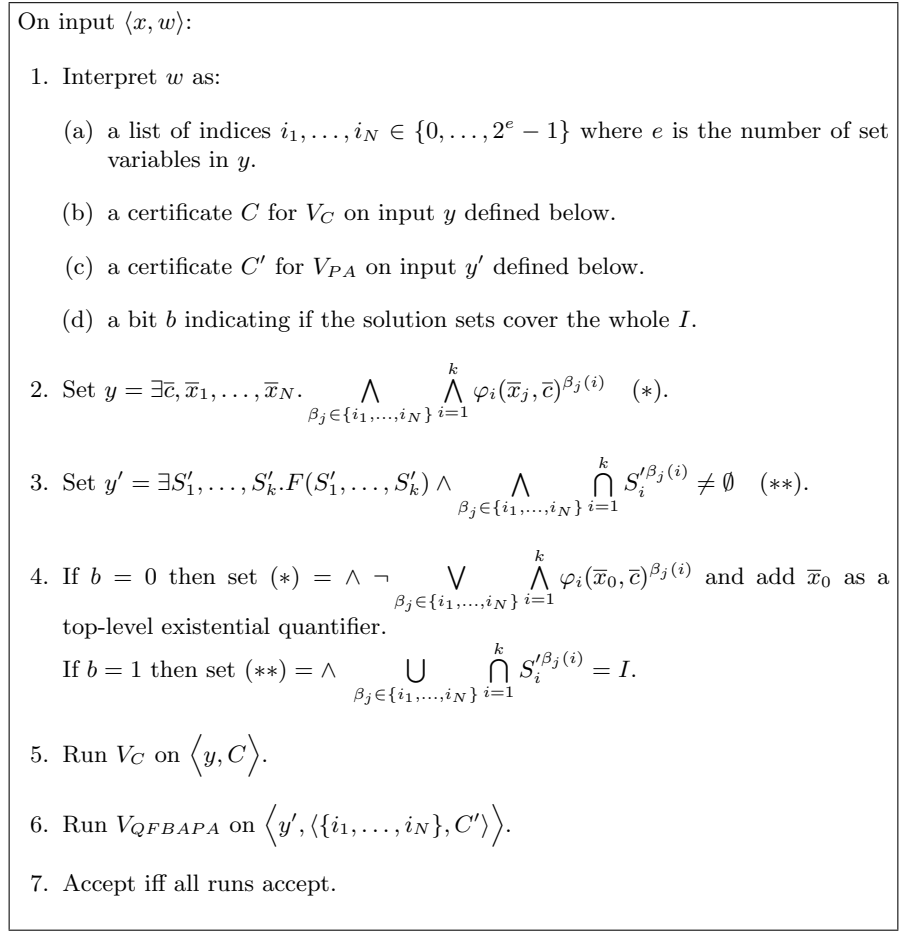

\fbox{\parbox{.95\textwidth}{
On input $\langle x, w \rangle$:

\begin{enumerate}
\setlength\itemsep{1em}
\item Interpret $w$ as:
\vspace{1em}
\begin{enumerate}
    \item a list of indices $i_1, \ldots, i_N \in \{0, \ldots, 2^e-1\}$ where $e$ is the number of set variables in $y$. 
    \item a certificate $C$ for $V_C$ on input $y$ defined below.
    \item a certificate $C'$ for $V_{PA}$ on input $y'$ defined below. 
    \item a bit $b$ indicating if the solution sets cover the whole $I$.
\end{enumerate}

\item Set $y = \exists \overline{c}, \overline{x}_1, \ldots, \overline{x}_N. \bigwedge\limits_{\beta_j \in \{i_1, \ldots, i_N \}}  \bigwedge\limits_{i = 1}^k \varphi_i(\overline{x}_j, \overline{c})^{\beta_j(i)} \hspace{1em} (*)$.

\item Set $y' = \exists S_1', \ldots, S_k'. F(S_1', \ldots, S_k') \land \bigwedge\limits_{\beta_j \in \{i_1, \ldots, i_N\}}  \mathop{\bigcap}\limits_{i = 1}^k S_i'^{\beta_j(i)} \neq \emptyset \hspace{1em} (**)$.

\item If $b = 0$ then set $(*) = \land \hspace{.5em}
\lnot \bigvee\limits_{\beta_j \in \{i_1, \ldots, i_N\}} \bigwedge\limits_{i = 1}^k \varphi_i(\overline{x}_0, \overline{c})^{\beta_j(i)}$ and add $\overline{x}_0$ as a top-level existential quantifier.
    
If $b = 1$ then set $(**) = \land \hspace{.5em} \bigcup\limits_{\beta_j \in \{i_1, \ldots, i_N\}} \bigcap\limits_{i = 1}^k S_i'^{\beta_j(i)} = I$.

\item Run $V_C$ on $\Big\langle y, C \Big\rangle$.

\item Run $V_{QFBAPA}$ on $\Big\langle y', \langle \{i_1, \ldots, i_N\}, C' \rangle \Big\rangle$.

\item Accept iff all runs accept.
\end{enumerate}
}}
\caption{Verifier for QFBAPA interpreted over index-sets.} 
\label{fig:intverif}
\end{figure*}    
\vspace{1em}

\noindent $\Rightarrow)$ If $x \in \mbox{QFBAPAI}$ then there exist $\overline{c}, \overline{s}$ satisfying: 
\[
F(S_1, \ldots, S_k) \land \bigwedge_{i = 1}^k S_i = \{ r \in I | \varphi_i(\overline{s}(r), \overline{c}) \}
\]
Define $S_i := \{ r \in I | \varphi_i(\overline{s}(r), \overline{c}) \}$. Then, the method of theorem \ref{qfbapa-complexity} applied to $F(S_1, \ldots, S_k)$ yields a formula $
G \land \bigwedge_{i = 0}^p |b_i| = k_i
$. Using $|b_i| = \sum\limits_{\beta \models b_i} \Big| \bigcap\limits_{i = 1}^k S_i^{\beta(i)}\Big|$ and setting $p_{\beta} := \bigcap\limits_{i = 1}^k S_i^{\beta(i)}, l_{\beta} := \Big|p_{\beta}\Big|$, yields $
G \land \bigwedge\limits_{i = 0}^p \sum\limits_{j = 0}^{2^e-1} \llbracket b_i \rrbracket_{\beta_j} \cdot l_{\beta_j} = k_i
$. Remove those $\beta$ where $l_{\beta} = 0$. Since:
\[
p_\beta = \bigcap\limits_{i = 1}^k \{r \in I| \varphi_i(\overline{s}(r), \overline{c}) \}^{\beta(i)} = \Bigg\{r \in I \Bigg| \mathop{\bigwedge}\limits_{i = 1}^k \varphi_i(\overline{s}(r), \overline{c})^{\beta(i)}\Bigg\}
\]
this includes those $\beta$ such that $\mathop{\bigwedge}\limits_{i = 1}^k \varphi_i(\overline{s}(r), \overline{c})^{\beta(i)}$ is not satisfiable. We obtain a reduced set of indices $\mathcal{R} \subseteq \{0, \ldots, 2^e-1\}$ where $
G \land \bigwedge\limits_{i = 0}^p \sum\limits_{\beta \in \mathcal{R}} \llbracket b_i \rrbracket_{\beta} \cdot l_{\beta} = k_i
$. Eisenbrand-Shmonin's theorem yields a polynomial family of indices such that $
G \land \bigwedge\limits_{i = 0}^p \sum\limits_{\beta \in \{i_1, \ldots, i_N\} \subseteq \mathcal{R}} \llbracket b_i \rrbracket_{\beta} \cdot l_{\beta}' = k_i
$ for non-zero $l_{\beta}'$. 

\vspace{.5em}

For each $\beta \in \{i_1, \ldots, i_N\}$, since $l_\beta \neq 0$, there exists $r_\beta \in I$ such that $
\bigwedge\limits_{i = 1}^k \varphi_i(\overline{s}(r_\beta), \overline{c})^{\beta(i)}
$. So the formula $y$ without (*) is satisfied.

\vspace{.5em}

The satisfiability  of the cardinality restrictions on $l_\beta'$ implies the existence of sets of indices $S_i'$ such that for each $\beta  \in \{i_1, \ldots, i_N\}$, $|p_\beta'| = l_\beta'$. Observe that $|I| = \sum\limits_{\beta \in \mathcal{R}} l_{\beta}$. Distinguish two cases:

\begin{itemize}
\item If $|I| > \sum\limits_{\beta \in \{i_1, \ldots, i_N\} } l_{\beta}'$ then there is at least one index $r_0$ such that $\overline{s}(r_0)$ satisfies $
\bigwedge\limits_{i = 1}^k \varphi_i(\overline{s}(r_0), \overline{c})^{\beta(i)}
$ for $\beta \notin \{i_1, \ldots, i_N\}$. Therefore, the formula $y$ with (*) is satisfied. In this case, define: \[
\overline{s}'(r) = 
\begin{cases}
\overline{s}(r_{\beta}) & \text{ if } r \in p_\beta' \text{ and } \beta \in \{i_1,\ldots,i_N\} \\
\overline{s}(r_0)         & \text{otherwise} 
\end{cases}
\]
and choose $b = 0$.

\item If $|I| = \sum\limits_{\beta \in \{i_1, \ldots, i_N\} } l_{\beta}'$ then define: 
\[
\overline{s}'(r) = 
\begin{cases}
\overline{s}(r_{\beta}) & \text{ if } r \in p_\beta'  \text{ and } \beta \in \{i_1,\ldots,i_N\}
\end{cases}
\] 
Here we choose $b = 1$. 
\end{itemize}

In any case, the formula $y$ that $V_C$ receives as input is satisfied. Since $N$ is polynomial in $|x|$, this gives a polynomially-sized certificate $C$ such that $V_C$ accepts $\langle y,C \rangle$ in polynomial time. 

\vspace{.5em}

Let $S_i'' = \{ r \in I |  \varphi_i(\overline{s}'(r), \overline{c}) \}$. Then $S_1'', \ldots, S_k''$ satisfy $y'$ by construction:

\begin{itemize}
\setlength\itemsep{1em}
\item Observe that for each $\beta \in \{i_1, \ldots, i_N\}$, $p_\beta'' = p_\beta'$.
\item For each $\beta \in \{i_1, \ldots, i_N\}$, 
$p_\beta''  \neq \emptyset$, since $l_\beta' \neq 0$.

\item If $b = 1$ then $\bigcup\limits_{\beta \in \{i_1, \ldots, i_N\}} p_\beta'' = I$ since $|I| = \sum\limits_{\beta \in \{i_1, \ldots, i_N\} } l_{\beta}'$.

\item The cardinality restrictions are satisfied by definition.
\end{itemize}

Again, since $N$ is polynomial in $|x|$, $|y'|$ is polynomial in $|x|$ too. By the above, it is also satisfiable. Thus, there exists a polynomially-sized certificate $C'$ for $V_{PA}$ such that $V_{QFBAPA}$ accepts $\langle \{i_1,\ldots,i_N\},C' \rangle
$ in polynomial time. So $V$ accepts $\langle x, \langle \{i_1, \ldots, i_N \}, C, C', b \rangle \rangle$ in polynomial time.

\vspace{1em}

$\Leftarrow)$ If $V$ accepts $\langle x, w \rangle$ in polynomial time then: 

\vspace{.5em}

\begin{itemize}
\item $\Big\langle y, C \Big\rangle$ is accepted by $V_C$, so there is a tuple $\overline{c}$ and for each $\beta \in \{i_1, \ldots, i_N\}$, there are tuples $s_\beta$, such that $\bigwedge\limits_{i = 1}^k \varphi_i(s_\beta(1), \ldots, s_\beta(n), \overline{c})^{\beta(i)}$.
\item $\langle y', \langle \{i_1, \ldots, i_N\}, C' \rangle \rangle$ is accepted by $V_{QFBAPA}$, so there exist sets $S_i'$ such that: \[
F(S_1', \ldots, S_k') \land \bigwedge\limits_{\beta \in \{i_1, \ldots, i_N\}} \mathop{\bigcap}\limits_{i = 1}^{k} S_i'^{\beta(i)} \neq \emptyset
\] 
\end{itemize}

\vspace{1em}
    
Interpreting $S_i'$ as index sets, we define an array $\overline{s}$ distinguishing two cases:

\vspace{.5em}

\begin{itemize}
\item If $b = 0$ then $V_C$ accepts: 
\[
\Big\langle \exists \overline{c}, \exists \overline{x}_1, \ldots, \overline{x}_N, \overline{x}_0.  \ldots 
 \lnot \bigvee\limits_{\beta \in \{i_1, \ldots, i_N\}} \bigwedge\limits_{i = 1}^k \varphi_i(\overline{x}_0, \overline{c})^{\beta(i)},
C \Big\rangle
\]
Let $s_0$ be a satisfying tuple for $\overline{x}_0$. Define: 
\[
\overline{s}(r) =
\begin{cases}
s_\beta & \text{ if } r \in p_\beta' \text{ and } \beta \in \{i_1, \ldots, i_N\} \\
s_0    & \text{otherwise}
\end{cases}
\]

\item If $b = 1$ then $S_i'$ satisfies $\bigcup\limits_{\beta \in \{i_1, \ldots, i_N\}} \bigcap\limits_{i = 1}^k S_i'^{\beta(i)} = I$. Define: 
\[
\overline{s}(r) =
\begin{cases}
s_\beta & \text{ if } r \in p_\beta' \text{ and } \beta \in \{i_1, \ldots, i_N\}
\end{cases}
\]
\end{itemize}

\vspace{.5em}

Then, by construction, $\overline{c}, \overline{s}$ form a solution of:
\[  
    \exists \overline{c}, \overline{x}.
    F(S_1,\ldots,S_k) \land \bigwedge_{i = 1}^k 
       S_i = \{ r \in I \mid \varphi_i(\overline{x}(r), \overline{c}) \}
\]
For each $\beta \in \{i_1, \ldots, i_N\}$: \[
p_{\beta} = \Big\{ r  \in I \Big| \bigwedge\limits_{i = 1}^k \varphi(\overline{s}(r), \overline{c})^{\beta(i)} \Big\} =  p_{\beta'}
\] so the cardinality    conditions are met.

\end{proof}

%% file: combination.tex
In this section we show, through a syntactic translation, that the conventional and generalized array 
operations can be expressed in QFBAPAI.
The combinatory array logic fragment of de Moura and Bj\o{}rner \cite{de_moura_generalized_2009} can be presented as a multi-sorted structure: 
\[
\mathcal{A} = \langle A,I,V, \cdot [ \cdot ], \text{store}(\cdot, \cdot, \cdot), \{c_i^v\}, \{f_i^v\}, \{R_i^v\}, \{c_j\}, \{f_j\}, \{R_j\} \rangle
\]
where $\mathcal{V} = \langle V, \{c_i^v\}, \{f_i\}^v, \{R_i^v\}\rangle$ is the structure modelling array elements and $I$ is a non-empty set which parametrizes the read ($\cdot[\cdot]$) and store ($store(\cdot, \cdot, \cdot)$) operations. Finally, $\Pi = \mathcal{V}^I = \langle A, \{c_j\}, \{f_j\}, \{R_j\} \rangle$ is the power structure with base $\mathcal{V}$ and index set $I$. Note that, according to the definition of a power structure, there is a one to one correspondence between the symbols of the component language and those of the array language. We use the superscript $v$ to distinguish between value symbols and power structure symbols. 
The read and store operations use a mixture of sorts. The read operation corresponds to a parametrized version of the canonical projection homomorphism of product structures \cite{hodges_model_1993}. It is interpreted as:
\begin{alignat*}{3}
  \cdot[\cdot]: & A \times I & \longrightarrow & V  \\
  & (a,i) & \longmapsto & a(i)
\end{alignat*}
On the other hand, the store
operation lacks a canonical counterpart in model theory. It is to be interpreted as the function: 
\begin{alignat*}{3}
  store: & A \times I \times V & \longrightarrow & A  \\
  & (a,i,v) & \longmapsto & store(a,i,v)
\end{alignat*}
where: 
\[
store(a,i,v)(j) =
\begin{cases}
a(j) & \text{if } j \neq i \\
v    & \text{if } j = i
\end{cases}
\]  
The goal of this section is to give a satisfiability preserving translation from CAL to QFBAPAI in such a way that the size of the transformed formula is bounded by a polynomial in the size of the original input. Since CAL formulas cannot express equicardinality constraints, $|A| = |B|$, this means that we have increased the expresive power of the fragment while retaining the same complexity bound. The translation is written in terms of a list of basic primitives explained below. The complete translation is shown in figure \ref{fig:translation}.

Since we are dealing with quantifier-free formulas, we map the propositional structure to boolean operations and concentrate in the encoding of non-propositional symbols. These symbols are atomic relations in either the component theory or the array theory. 

\textbf{Relations in the component theory}. An atomic formula of the component theory has the following shape:
\[
R^v(f_i\{a_1[i_1], \ldots, a_n[i_n], c_1, \ldots, c_m \})
\]
Here and in the rest of the section we use the notation $R(f_i\{p_1, \ldots, p_n\})$ for a list of $\text{arity}(R)$ function terms of the form $f_i\{p_1, \ldots,p_n\}$ where $f_i$ is a function symbol using a subset of the parameters in $\{p_1, \ldots, p_n\}$. Both $f_i$ and the parameters $p_i$ must have the same sort as $R$. We use the letter $a$ to denote either an array variable $x$ or a $store$ term and the letter $v$ to denote an element value in contrast to a read term $a[j]$.

\vspace{.5em}

We transform the above constraint using the following rules:

\begin{enumerate}
    \item ABSTRACT READS ($\le$ 1): if there are more than two parameters that use the read function $\cdot[\cdot]$ applied to a variable, we rewrite all occurrences $x_j[i]$ but one into value constants $x_{ji}$. Note that a read from a constant array need not create a new value variable. Instead, we rewrite $c[i]$ as $c^v$. In this case, no further changes are required in later stages.
    \item IMPOSE READS: for each abstracted read $x_j[i]$ add the condition:
    \[
    \{l \in I|x_j(l) = x_{ji}\} \supseteq \{i\}
    \]
    \item ABSTRACT WRITES: rewrite the innermost store operations $store(x,i,v)$ into array variables $x_{iv}$.
    \item IMPOSE WRITES: for each abstracted store $x_{iv}$, we impose the condition:
    \[
    \{l \in I|x_{iv}(l) = v\} \supseteq \{i\} \land 
    \{l \in I|x_{iv}(l) = x(l)\} \supseteq \{i\}^c
    \]
\end{enumerate}

This process is repeated until there is no change in the manipulated formula. In this last case, we have obtained a relation:
\[
R^v(f_i\{x[i], \text{abs}_1, \ldots, \text{abs}_k, c_1,\ldots, c_m\})
\]
where $\text{abs}_j$ are the newly introduced array or value variables. We then perform one last step:

\begin{enumerate}[resume]
\item IMPOSE VALUE CONSTRAINT: add the constraint:
\[
\{l \in I|R^v(f_i\{x(l), \text{abs}_1, \ldots, \text{abs}_k, c_1,\ldots, c_m\})\} \supseteq \{i\}
\]
\end{enumerate}

\begin{figure*}[ht!]
\fbox{\parbox{.95\textwidth}{
Given a formula $\psi$ of CAL in negation normal form:

\begin{enumerate}
\setlength\itemsep{1em}
\item Rewrite $\land \mapsto \cap, \lor \mapsto \cup$ and $\lnot \mapsto \cdot^c$.
     
\item Consider the following auxiliary procedure $P$ receiving one bit $b$ as parameter.

\vspace{.5em}

Repeat until no more constraints are added:

\vspace{.5em}

\begin{enumerate}
    \setlength\itemsep{.25em}
    \item $\begin{array}{ccc}
    \text{If } \, b = 0 & \text{ then } & \text{ABSTRACT READS} (= 0) \\
    & \text{else} & \;\; \text{ABSTRACT READS} (<= 1).
    \end{array}$

    \item IMPOSE READS
    \item ABSTRACT WRITES
    \item IMPOSE WRITES
\end{enumerate}

\item For each relation in the array theory call $P$ with $b = 0$.

\item For each relation in the component theory call $P$ with $b = 1$.
\end{enumerate}
}}
\caption{Translation scheme from CAL to QFBAPAI.}
\label{fig:translation}
\end{figure*}

\textbf{Relations in the power structure theory.} An atomic formula of the product theory has the shape:
\[
R(f_i\{a_1, \ldots, a_n, c_1, \ldots, c_m\})
\]
where $c_1, \ldots, c_m$ are to be interpreted as constants of the product. We repeat a variation of the steps 1-4 where ABSTRACT READS ($\le$ 1) is changed into ABSTRACT READS (= 0). The only difference between the two is that the latter removes all reads. The result of this operation is a relation:
\[
R(f_i\{x_1,\ldots,x_s,\text{abs}_1, \ldots, \text{abs}_k, c_1, \ldots, c_m\})
\]
where $\text{abs}_j$ are the newly introduced array variables. We cannot have value variables since in this case value expressions are not top-level.

\vspace{.5em}

In this case, we do the following as a last step:

\begin{enumerate}
  \setcounter{enumi}{4}
  \item IMPOSE ARRAY CONSTRAINT: add the constraint:
  \[
  \{l \in I|R(f_i\{x_1(l), \ldots, x_s(l), \text{abs}_1(l), \ldots, \text{abs}_k(l), c_1, \ldots, c_m\})\} = I
  \]
\end{enumerate}

\textbf{Satisfiability preservation and size of the transformed formula}. It is clear that each transformation step yields an equisatisfiable formula. In particular, this ensures that the order of introduction of new variables does not matter. Even if the transformed formula may contain duplicates, the existence of a solution is equivalent in both formulas.

Regarding the size of the transformed formula, we observe that during the analysis of a relation we create as many variables as the size of such relation. Thus, the number of variables created is at most linear in the size of the formula. This means that the total number of variables and constants that are either present in the original formula or created by the algorithm, $C$, is in $O(|\psi|)$.

The creation of each variable implies the creation of at most three restrictions: this happens in the IMPOSE WRITES case, where the third restriction specifies that the size of $\{i\}$ is one. Each restriction uses at most two variables, so we can encode it using $O(log_2(|\psi|))$ space. Thus, to encode all the added restrictions we need $O(|\psi| \log_2(|\psi|))$ space.

Each relation generates an additional constraint, which may use all the set of $C$ variables. So we may need up to $O(|\psi| \cdot \log_2(|\psi|))$ to encode the constraint. Since there are $O(|\psi|)$ relations, we need $O(|\psi|^2 \log_2(|\psi|))$ space to encode them.

Overall, the size increase is in $O(|\psi|^2 \log_2(|\psi|)$, as desired to preserve NP complexity.

%% file: related.tex
Our work is related to a long tradition of decision procedures for the theories of arrays \cite{mccarthy_towards_1993, stump_decision_2001, bradley_whats_2006, daca_array_2016, habermehl_what_2008, ghilardi_decision_2007}. Our direct inspiration is combinatory array logic \cite{de_moura_generalized_2009}. We have extended this fragment with cardinality constraints.

In our study, we have given priority to those procedures that decide satisfiability within the NP complexity class. From these, \cite{alberti_cardinality_2017} and \cite{daca_array_2016} are the more closely related since they also address counting properties. The main difference with these works is that our index theory is arbitrary and that the element theory is any one in NP. This gives access to a greater degree of compositionality. For instance, we can profit of the properties of QFBAPA to handle infinite cardinalities in the index theory \cite{kuncak_ordered_2010}. On the other hand, the work of \cite{daca_array_2016} allows for  a great expressivity, achieving NP complexity on particular fragments, but it is PSPACE-complete in the general case.

Other influential works in the theory of integer arrays include \cite{bradley_whats_2006} and \cite{habermehl_what_2008}. \cite{bradley_whats_2006} treats a fragment capable of expressing ordering conditions and Presburger restrictions on the indices. \cite{habermehl_what_2008} complements the work above based on automata considerations. In both cases, the complexity of the satisfiability problem for the full fragment remains, to our knowledge, open. Parametric theories of arrays include \cite{stump_decision_2001, de_moura_generalized_2009} and \cite{alberti_decision_2015}. However, the line of work in \cite{alberti_decision_2015} as consolidated in the doctoral thesis \cite{alberti_smt-based_2015}, only shows decidability and NEXPTIME completeness on particular instances. None of \cite{stump_decision_2001, bradley_whats_2006, habermehl_what_2008, de_moura_generalized_2009, alberti_decision_2015, alberti_smt-based_2015} treat cardinality constraints. 

%% file: conclusion.tex
We have identified the model theoretic structure behind a state of the art fragment of the theory of arrays. We have given self-contained proofs of complexity which shed light on the underlying constraints that the fragment addresses. This has allowed to generalize the fragment to encode arbitrary cardinality constraints. Our work also shows that the set variables of BAPA can be interpreted to encode useful restrictions.

As future work, we plan to build on the efforts in \cite{de_moura_generalized_2009}, to provide an efficient implementation of the fragment. We would also like to perform a cross-fertilization with other fragments of the theory of arrays providing counting capabilities, while exploring the interactions between their seemingly different foundations.